\documentclass[11pt]{extarticle}
\usepackage[a4paper,
            left=2cm,
            right=2cm,
            top=2cm,
            bottom=2cm
           ]{geometry}

\usepackage{comment}
\usepackage{siunitx}
\usepackage{relsize}
\usepackage{ifthen}
\usepackage{caption}

\usepackage{color}
\usepackage[T1]{fontenc}
\usepackage{psfrag}
\usepackage{booktabs}
\usepackage{mathtools}
\usepackage{amsthm}
\usepackage{xstring}
\usepackage{multirow}
\usepackage{xcolor}
\usepackage{prettyref}
\usepackage{flexisym}
\usepackage{bigdelim}
\usepackage{breqn} 
\usepackage{listings}

\usepackage{enumitem}
\usepackage{xspace}
\usepackage{bm}
\graphicspath{{./figures/}}
\usepackage{tikz}
\usetikzlibrary{matrix,calc}
\usepackage{amssymb}  
\usepackage{mdframed}
\usepackage[vlined,ruled,linesnumbered]{algorithm2e}

\newtheorem{theorem}{Theorem}

\newtheorem{lemma}[theorem]{Lemma}
\newtheorem{definition}{Definition}
\newtheorem{proposition}[theorem]{Proposition}

\newcommand{\bdmath}{\begin{dmath}}
\newcommand{\edmath}{\end{dmath}}
\newcommand{\beq}{\begin{equation}}
\newcommand{\eeq}{\end{equation}}
\newcommand{\bdm}{\begin{displaymath}}
\newcommand{\edm}{\end{displaymath}}
\newcommand{\bea}{\begin{eqnarray}}
\newcommand{\eea}{\end{eqnarray}}
\newcommand{\beal}{\beq \begin{array}{ll}}
\newcommand{\eeal}{\end{array} \eeq}
\newcommand{\beas}{\begin{eqnarray*}}
\newcommand{\eeas}{\end{eqnarray*}}
\newcommand{\ba}{\begin{array}}
\newcommand{\ea}{\end{array}}
\newcommand{\bit}{\begin{itemize}}
\newcommand{\eit}{\end{itemize}}
\newcommand{\ben}{\begin{enumerate}}
\newcommand{\een}{\end{enumerate}}

\newcommand{\ie}{\emph{i.e.,}\xspace}

\newcommand{\hide}[1]{}

\newcommand{\hiddenText}{{\color{gray} hidden text.}}
\newcommand{\hideWithText}[1]{\hiddenText}

\newcommand{\tran}{^{\mathsf{T}}}

\newcommand{\Real}[1]{ { {\mathbb R}^{#1} } }
\newcommand{\Q}[1]{ { {\mathbb Q}^{#1} } }
\newcommand{\Z}[1]{ { {\mathbb Z}^{#1} } }

\newcommand{\bmat}{\left[ \begin{array}}
    \newcommand{\emat}{\end{array}\right]}
\newcommand{\blue}[1]{{\color{blue}#1}}

\newcommand{\linkToPdf}[1]{\href{#1}{\blue{(pdf)}}}
\newcommand{\linkToPpt}[1]{\href{#1}{\blue{(ppt)}}}
\newcommand{\linkToCode}[1]{\href{#1}{\blue{(code)}}}
\newcommand{\linkToWeb}[1]{\href{#1}{\blue{(web)}}}
\newcommand{\linkToVideo}[1]{\href{#1}{\blue{(video)}}}
\newcommand{\linkToMedia}[1]{\href{#1}{\blue{(media)}}}
\newcommand{\award}[1]{\xspace} 

\usepackage{scalerel,stackengine}
\newcommand\pig[1]{\mathlarger{#1}}
\newcommand\pigl[1]{\mathopen{\pig{#1}}}
\newcommand\pigr[1]{\mathclose{\pig{#1}}}

\newcommand{\lovasz}{Lov\a' asz\space}
\newcommand{\tr}{\operatorname{Tr}}
\newcommand{\symm}{\mathbb{S}^n}

\newcommand{\ip}[2]{\left\langle #1,\,#2 \right\rangle}

\usepackage{graphicx} 
\usepackage[ruled,vlined]{algorithm2e}
\usepackage{hyperref}
\usepackage{tikz}
\usetikzlibrary{graphs, graphs.standard, matrix}

\title{A Polynomial-Time Algorithm for Coloring Perfect Graphs\\ Based on Walk Counting}
\author{Amir Ali Ahmadi\thanks{Princeton University, Operations Research and Financial Engineering. Partially supported by the Princeton AI Lab Seed Grant, the Princeton SEAS Innovation Grant, and a
Research Gift in Mathematical Optimization. Email: \texttt{aaa@princeton.edu, yt3846@princeton.edu}},
\and 
Pravesh K. Kothari\thanks{Princeton University, Computer Science. Partially supported by the Princeton AI Lab Seed Grant and the Princeton SEAS Innovation Grant. Email: \texttt{kothari@cs.princeton.edu}},
\and 
Yukai Tang\footnotemark[1]
}

\begin{document}

\maketitle

\begin{abstract}

We present a polynomial-time algorithm for optimally coloring perfect graphs that is based entirely on graph-theoretic operations. At its core, the algorithm decides whether a perfect graph contains a clique of a given size by iteratively counting walks in the graph with certain weights assigned to its edges and nonedges. These weights are initialized according to a uniform scheme and then updated in each iteration based on the walk counts from the previous iteration. 

\end{abstract}

\section{Introduction}

In a simple graph $G(V,E)$, with vertex set $V$ and edge set $E$, a \emph{stable set} is a set of pairwise non-adjacent vertices and a \emph{clique} is a set of pairwise adjacent vertices. A \emph{coloring} of $G$ is an assignment of colors to the vertices of $G$ such that no two adjacent vertices share the same color. This can also be viewed as partitioning $V$ into stable sets. We denote the size of a maximum stable set of $G$ by $\alpha(G)$, the size of a maximum clique of $G$ by $\omega(G)$, and the minimum number of colors needed to color $G$ by $\chi(G)$. 

In this paper, we are concerned with finding a maximum stable set, a maximum clique, and a minimum coloring of a \emph{perfect graph}. A graph $G$ is called \emph{perfect} if for every induced subgraph\footnote{A graph \(H\) is an \emph{induced subgraph} of a graph \(G\) if \(V(H)\subseteq V(G)\) and any two vertices of \(H\) are adjacent if and only if they are adjacent in \(G\).} $H$ of $G$, the chromatic number $\chi(H)$ is equal to the clique number $\omega(H)$. This is a highly studied family of graphs both in structural graph theory and in combinatorial optimization.  One of the most significant results regarding perfect graphs is the strong perfect graph theorem conjectured by Berge~\cite{berge1961farbung} and proven many years later by Chudnovsky, Robertson, Seymour, and Thomas~\cite{chudnovsky2006strong}. This theorem states that a graph $G$ is perfect if and only if no induced subgraph of $G$ is an odd cycle of length at least five or the complement of one. Among other implications, this theorem has led to a polynomial-time algorithm for recognizing perfect graphs~\cite{chudnovsky2005recognizing}.




It is well known that the problems of finding a maximum stable set, a maximum clique, and a minimum coloring of a general graph are all NP-hard~\cite{karp2009reducibility}. In an influential paper, Gr\"otschel, Lov\a' asz, and Schrijver proved that these problems can be solved in polynomial time for perfect graphs~\cite{Grotschel1994poly-alg-perfect-graphs}. Their approach, however, relies on the ellipsoid method and is not regarded as combinatorial.
Designing a combinatorial polynomial-time algorithm for finding a maximum stable set/clique or a minimum coloring of a perfect graph is a well-known open problem in graph theory; see, e.g., an interview with Lov\a' asz~\cite{lovasz2011interview}, an interview with Chudnovsky~\cite{Mansour2021ChudnovskyInterview},~\cite[Chapter 9]{grotschel2012geometric-alg-combo-opt}, \cite[Section 12]{trotignon2013perfect-graphs-survey},~\cite[Section 1]{chudnovsky2020maximum}, or~\cite[Section 1]{abrishami2025submodular}. Over the years, such algorithms have been designed for several subsets of perfect graphs, including interval graphs~\cite{olariu1991optimal}, chordal graphs~\cite{gavril1972algorithms}, claw-free perfect graphs~\cite{hsu1981color}, bull-free perfect graphs~\cite{penev2012coloring}, perfect graphs that do not have a balanced skew partition~\cite{chudnovsky2015coloring}, and bounded degree perfect graphs with no prism or hole of length four as induced subgraphs~\cite{abrishami2025submodular}.
%
%
%
%
%


In this paper, we present a new polynomial-time algorithm for finding a maximum stable set, a maximum clique, and a minimum coloring in any perfect graph. At its core, our algorithm decides if a perfect graph contains a clique of a given size. This procedure relies solely on counting walks in the graph with certain weights assigned to its edges and nonedges. These weights are initially set in a uniform fashion, and then updated in each iteration using the walk counts from the previous iteration.


Despite the prominence of the open problem on finding a polynomial-time combinatorial algorithm to optimally color perfect graphs, there is surprisingly no universally accepted definition of a combinatorial algorithm~\cite{trotignon2013perfect-graphs-survey, chudnovsky2020maximum,abrishami2025submodular,eisenbrand2003combo-max-stable-t-perfect,schrijver2000combinatorial,IWATA2002203}. 
Roughly speaking, a `combinatorial algorithm' should be an algorithm that can be described as a sequence of operations  applied to vertices and edges of graphs. 
In the absence of a formal definition of the notion of a combinatorial algorithm, it is hard for us to genuinely claim that we have solved the open problem. We believe that walk counting is a basic combinatorial operation and hence our algorithm should be considered combinatorial. We hope that either our work will be accepted by the community as resolving this problem, or that it will lead to a formal definition of the notion of a combinatorial algorithm that excludes our algorithm. Either outcome would constitute meaningful progress on this longstanding question.

\subsection{Organization and contributions of the paper}

%

{

The remainder of this paper is organized as follows. Section~\ref{sec:reduction-color-clique-decision} details a {polynomial-time} reduction from {the problem of optimally coloring a perfect graph to that of deciding if a perfect graph has a clique of a given size.}
%
%
In Section~\ref{sec:certificate-for-refuting-clique-of-size-k}, {we prove that one can always give a certificate of nonexistence of a clique of a given size in a perfect graph in terms of an inequality involving the weights of closed walks of a certain length (Theorem~\ref{thm:inexact-certificate-for-refuting-clique-of-size-k}). In Section~\ref{sec:alg-analysis}, we present our walk counting algorithm for deciding if a perfect graph has a clique of a given size. The correctness of this algorithm is established by a Lyapunov function argument (Section~\ref{sec:proof_of_correctness}). {Finally, in Section~\ref{sec:connection}, we show that our algorithm can be viewed as a combinatorial (approximate) implementation of a multiplicative weights update algorithm for solving a modified version of the semidefinite program associated with the \lovasz theta function, or a combinatorial (approximate) implementation of the exponentiated gradient descent algorithm applied to the dual of this semidefinite program. }


}

\subsection{Notation}\label{sec:notations}
For a graph $G$, we denote its complement by $\bar{G}$.
We denote the set of $n\times n$ real symmetric matrices by $\symm$.
A matrix $A\in \symm$ is said to be positive semidefinite (denoted by $A\succeq 0$) if all its eigenvalues are nonnegative. We denote the identity matrix by $I$, the all-ones matrix by $J$, the maximum eigenvalue of a matrix $A\in\symm$ by $\lambda_\max(A)$, its minimum eigenvalue by $\lambda_\min(A)$, and its trace by $\tr(A)$. 

For a vector $x\in \Real{n}$, we denote its transpose by $x\tran$ and its $i$-th entry by $x_i$. The $i$-th standard basis vector in $\Real{n}$ is denoted by $e_i$. We denote the probability simplex in $\Real{n}$ by $\Delta_n$; any element of $\Delta_n$ is called a probability vector. We use the notation $Y_{ij}$ to denote the symmetric matrix that has ones in entries $(i,j)$ and $(j,i)$, and zeros elsewhere. We denote the indicator function of a logical condition $x$ by $\mathbf{1}(x)$; i.e., $\mathbf{1}(x) = 1$ if $x$ is true and $\mathbf{1}(x) = 0$ otherwise.


%

{
}

\section{Reducing minimum coloring of a perfect graph to clique decision}\label{sec:reduction-color-clique-decision}

In this section, we give a polynomial-time combinatorial reduction from the problem of finding a minimum coloring of a perfect graph to that of deciding if the graph has a clique of a given size. While the main part of this reduction appears in slightly different forms in the literature~\cite{grotschel2012geometric-alg-combo-opt,trotignon2013perfect-graphs-survey,laurent2012semidefinite-lecture-notes}, we give a self-contained presentation here with some details filled in for completeness and the benefit of the reader.


Suppose we have an oracle that, given a perfect graph $G(V,E)$ and an integer $k\in\{2,\dots,|V|\}$, determines whether there is a clique of size $k$ in $G$. In the analysis below, we will refer to this oracle as the clique decision oracle.
The reduction is as follows. We first show how to find a maximum clique in a perfect graph using the clique decision oracle (Algorithm~\ref{alg:blackbox2}).\footnote{This procedure works more generally on all graphs if given access to a clique decision oracle for an arbitrary graph.} Then, we show how Algorithm~\ref{alg:blackbox2} can be used to find a stable set that intersects a given collection of maximum cliques in a perfect graph (Algorithm~\ref{alg:blackbox1}). Finally, we show how Algorithm~\ref{alg:blackbox2} and Algorithm~\ref{alg:blackbox1} can be combined to color a perfect graph (Algorithm~\ref{alg:find-critical-stable-set}).

\subsection{Searching for a maximum clique using the clique decision oracle}

In this section, we show how to find a maximum clique in a perfect graph using the clique decision oracle. The procedure is outlined in Algorithm~\ref{alg:blackbox2}. 


We begin by computing the clique number $\omega(G)$ by binary search over $k \in \{2,\dots,|V|\}$ using the clique decision oracle. We then fix an ordering of the vertices of $G$ and maintain a working graph $H$, initially equal to $G$. We process the vertices in the chosen order. For every vertex $v$, we use the clique decision oracle to decide whether the (still perfect) graph $H \setminus \{v\}$ contains a clique of size $\omega(G)$. If the answer is yes, then there exists a maximum clique of $H$ that does not contain $v$, so $v$ can be deleted safely. If the answer is no, then every clique of size $\omega(G)$ in $H$ contains $v$, so $v$ must be kept. Since a vertex is deleted only when the clique number remains $\omega(G)$, the current working graph always has clique number $\omega(G)$. At the end of the procedure, every remaining vertex belongs to every maximum clique of the final working graph $H$. Consequently, $H$ itself is a maximum clique of $G$.

\vspace{2mm}
\begin{algorithm}[H]
\caption{Finding a maximum clique using a clique decision oracle}\label{alg:blackbox2}
\SetAlgoLined
\KwIn{A perfect graph $G(V,E)$}
\KwOut{A maximum clique of $G$}

Compute $\omega(G)$ using binary search and the clique decision oracle\;
$H \leftarrow G$\;
Fix any ordering $v_1,\dots,v_{|V|}$ of $V$\;

\For{$i=1$ \KwTo $|V|$}{
    \If{$H\setminus\{v_i\}$ contains a clique of size $\omega(G)$}{
        $H \leftarrow H\setminus\{v_i\}$\;
    }
}

\Return the vertex set of $H$
\end{algorithm}



\subsection{Finding a stable set that intersects a given set of maximum cliques}

In this section, we show how to find a stable set in a perfect graph $G$ that intersects a given collection of maximum cliques $\mathcal{C} = \{C_1, C_2, \ldots, C_m\}$  of $G$. This is done as given in Algorithm~\ref{alg:blackbox1}. 

\vspace{2mm}
\begin{algorithm}[H]
\caption{Finding a stable set that intersects every maximum clique in a given set}\label{alg:blackbox1}
\SetAlgoLined
\KwIn{A perfect graph $G=(V,E)$, a collection of maximum cliques $\mathcal{C}=\{C_1,\dots,C_m\}$ of $G$}
\KwOut{A stable set $S$ with $S\cap C_i\neq\emptyset$ for all $i$}

\ForEach{$v\in V$}{
    $w(v)\leftarrow |\{\,i\mid v\in C_i\}|$\;
}

Construct an auxiliary graph $G'=(V',E')$ as follows: \newline
$V'$: for every vertex $v\in V$, delete $v$ and, if $w(v)>0$, add $w(v)$ copies
$\{v^{(1)},\dots,v^{(w(v))}\}$ to $V'$; \newline
$E'$: for every edge $uv\in E$, make every copy of $u$ in $V'$ adjacent to every copy of $v$ in $V'$\;

$S' \leftarrow \text{Algorithm~\ref{alg:blackbox2}}(\bar{G'})$\;
$S \leftarrow \{\,v\in V \mid \exists j \text{ such that } v^{(j)}\in S'\,\}$\;

\Return $S$
\end{algorithm}
\vspace{2mm}

The proof of correctness of Algorithm~\ref{alg:blackbox1} is given in Proposition~\ref{prop:blackbox1-correctness} and relies on the two lemmas below.

 \begin{lemma}\label{lem:stable-set-touch-all-cliques-exist}
Let $G$ be a perfect graph. Then there exists a stable set $S$ in $G$ such that $S$ intersects every maximum clique of $G$.
\end{lemma}

\begin{proof}
Suppose for the sake of contradiction that no stable set in $G$ intersects every maximum clique. Let
\(
S_1,S_2,\dots,S_{\chi(G)}
\)
be the color classes of a minimum coloring of $G$. Since $S_1$ is a stable set and, by assumption, does not intersect every maximum clique, there exists a maximum clique $K$ of $G$ such that
\(
K\cap S_1=\emptyset.
\)
It follows that \(K\subseteq G\setminus S_1\), and hence
\(
\omega(G\setminus S_1)\ge \omega(G).
\) Meanwhile, the remaining color classes
\(
S_2,\dots,S_{\chi(G)}
\)
form a valid coloring of \(G\setminus S_1\). Therefore,
\(
\chi(G\setminus S_1)\le \chi(G)-1.
\)

Since \(G\) is perfect, every induced subgraph of \(G\) is perfect; in particular, \(G\setminus S_1\) is perfect. Thus
\(
\chi(G\setminus S_1)=\omega(G\setminus S_1).
\)
Combining this with the inequalities above yields
\[
\omega(G)\le \omega(G\setminus S_1)=\chi(G\setminus S_1)\le \chi(G)-1,
\]
which contradicts the perfectness of \(G\). 
\end{proof}

 \begin{definition}[False twin copy]
Let $G=(V,E)$ be a graph, and let $v\in V$. A vertex $v'$ added to $G$ is called a \emph{false twin copy} of $v$ if $v'$ is adjacent to exactly the neighbors of $v$ and is not adjacent to $v$.
\end{definition}

\begin{lemma}\label{lem:perfectness-maintained}
    For any perfect graph $G$, the auxiliary graph $G'$ constructed in Algorithm~\ref{alg:blackbox1} is still perfect.
\end{lemma}

The proof of Lemma~\ref{lem:perfectness-maintained} follows immediately from Theorem~1 in~\cite{lovasz-substitute}. 

\begin{proposition}\label{prop:blackbox1-correctness}
Let $G=(V,E)$ be a perfect graph, and let $\mathcal{C}=\{C_1,C_2,\dots,C_m\}$ be a collection of maximum cliques in $G$. Then Algorithm~\ref{alg:blackbox1} returns a stable set $S$ in $G$ that intersects every maximum clique in $\mathcal{C}$.
\end{proposition}

\begin{proof}
We first prove that the output $S$ of Algorithm~\ref{alg:blackbox1} is a stable set in $G$. By Lemma~\ref{lem:perfectness-maintained}, the auxiliary graph $G'$ constructed {in} Algorithm~\ref{alg:blackbox1} is perfect. Hence, its complement $\bar G'$ is also perfect. {Therefore, by the guarantee of Algorithm~\ref{alg:blackbox2}, the set $S'$ in line 5 of Algorithm~\ref{alg:blackbox1} is a maximum clique of $\bar G'$, i.e., a maximum stable set of $G'$.} {By the construction of the auxiliary graph $G'$, the set $S$ in line 6 of Algorithm~\ref{alg:blackbox1} has the following equivalent description:
\[
S=\{\,v\in V : \text{some false twin copy of } v \text{ belongs to } S'\,\}.
\]}We claim that $S$ is a stable set in $G$.  {Suppose for the sake of contradiction that  $u,v\in S$ are adjacent in $G$. Since $u,v\in S$, there exist copies \(u^{(a)}\) of \(u\) and \(v^{(b)}\) of \(v\) that both belong to \(S'\). However, since $u,v$ are adjacent in $G$, by the construction of $G'$, \(u^{(a)}\) is adjacent to  \(v^{(b)}\). This contradicts the fact that $S'$ is a stable set in $G'$.}

It remains to {show} that $S$ intersects every clique in $\mathcal{C}$. {Suppose for the sake of contradiction} that there exists some $C_j\in\mathcal{C}$ such that
\(
S\cap C_j=\emptyset.
\)
Recall that the weight $w(v)$ {in Algorithm~\ref{alg:blackbox1}} is defined as the number of cliques in $\mathcal{C}$ that contain $v$. Therefore,
\(
\sum_{v\in S} w(v)
=
\sum_{i=1}^m |S\cap C_i|.
\)
{Since $S$ is a stable set and each $C_i$ is a clique, $S$ intersects every $C_i$ at at most one vertex. Together with the fact that $S\cap C_j=\emptyset$, it follows that}
\(
\sum_{v\in S} w(v)\le m-1.
\)
{
Since \(S'\) contains only copies of vertices in \(S\) and there are exactly \(w(v)\) copies of each \(v\), we have
\[
|S'|\le \sum_{v\in S} w(v) \le m-1.
\]}{From Lemma~\ref{lem:stable-set-touch-all-cliques-exist}}, there exists a stable set $T$ in $G$ that intersects every $C_i$. {We construct a stable set $T'$ in $G'$ by taking all $w(v)$ false twin copies of $v$ for each vertex $v\in T$. Since false twin copies of the same vertex are pairwise nonadjacent and the vertices of $T$ are pairwise nonadjacent in $G$, $T'$ is a stable set in $G'$.} Moreover,
\[
|T'|=\sum_{v\in T} w(v)=m > |S'|.
\]
This contradicts the fact that $S'$ is a maximum stable set in $G'$, which shows that $S$ must intersect every clique in $\mathcal{C}$.

\end{proof}







\subsection{Reducing coloring a perfect graph to clique decision}

Algorithms~\ref{alg:blackbox2} and~\ref{alg:blackbox1} together lead to a reduction from coloring a perfect graph to deciding whether the graph contains a clique of a given size. Algorithm~\ref{alg:find-critical-stable-set} implements this reduction.

\begin{algorithm}
\caption{Optimally coloring a perfect graph}\label{alg:find-critical-stable-set}
\SetAlgoLined
\textbf{Input:} A perfect graph $G(V,E)$

    {\bf Output:} A minimum coloring of $G$

\While{$G$ not empty}
{
Call Algorithm~\ref{alg:blackbox2} to get a maximum clique $C$ in $G$ 

Initialize the collection of maximum cliques $\mathcal{C}$ in $G$:  $\mathcal{C}\gets \{C\}$

\While{True}
{
Call {Algorithm~\ref{alg:blackbox1}} to get a stable set $S$ that intersects every maximum clique in $\mathcal{C}$\\
\If {$\omega(G) > \omega(G\setminus S)$}{
    \textbf{Break}
}

Call {Algorithm~\ref{alg:blackbox2}} to get a maximum clique $C'$ in $G\setminus S$\\ 
$\mathcal{C} \gets \mathcal{C}\cup\{C'\}$\\
}

Color $S$ with a new color

$G \gets G\setminus S$\\ 
}
\end{algorithm}

We now prove the following proposition.

\begin{proposition}
    For every perfect graph $G(V,E)$, Algorithm~\ref{alg:find-critical-stable-set} gives an optimal coloring of $G$.\footnote{It follows that the number of iterations of the outer loop is $\chi(G)$.} Moreover, during any execution of the inner loop, the number of iterations is at most $|V|$.
\end{proposition}

\begin{proof}
    {We first prove that Algorithm~\ref{alg:find-critical-stable-set} gives an optimal coloring of $G$.} Fix one outer iteration of Algorithm~\ref{alg:find-critical-stable-set}, and let $H$ denote the current graph at the start of that outer iteration. Let $S$ be the stable set produced when the inner loop terminates.  {When the inner loop terminates},
\(
\omega(H \setminus S) < \omega(H).
\)
We claim that $S$ intersects every maximum clique of $H$. Indeed, if there is a maximum clique $Q$ of $H$ disjoint from $S$, then $Q \subseteq H \setminus S$, and hence
\(
\omega(H \setminus S) \ge |Q| = \omega(H),
\)
contradicting the strict inequality above. 

Since $S$ is a stable set, it contains at most one vertex from any clique of $H$. {Together with the fact that $S$ intersects every maximum clique of $H$,} every maximum clique of $H$ loses exactly one vertex when $S$ is deleted. It follows that
\(
\omega(H \setminus S) = \omega(H)-1.
\)
We also know that $H \setminus S$ is an induced subgraph of the perfect graph $H$, so it is also a perfect graph. Consequently,
\[
\chi(H \setminus S)=\omega(H \setminus S)=\omega(H)-1=\chi(H)-1.
\]
Thus each outer iteration removes one stable set and decreases the chromatic number of the remaining perfect graph by exactly one. Starting from \(G\), the algorithm therefore performs exactly \(\chi(G)\) outer iterations and uses exactly \(\chi(G)\) colors.

It remains to bound the number of iterations of every inner loop. Fix again one outer iteration, and let $H(V_H,E_H)$ be the graph considered during that {outer} iteration. 
{At inner iteration \(t\), denote by \(\mathcal{C}_t=\{C_1,\dots,C_t\}\) the collection of maximum cliques accumulated so far.} Define the affine subspace
\[
L^t \coloneqq \bigl\{x\in\mathbb{R}^{|V_H|} : x(C_i)=1 \text{ for all } i=1,\dots,t\bigr\},
\]
where
\(
x(C_i)\coloneqq \sum_{v\in C_i} x_v.
\)
{Since \(L^{t+1}\) is obtained from \(L^t\) by adding the constraint \(x(C_{t+1})=1\), we have \(L^{t+1}\subseteq L^t\).} Let $S_t$ be the stable set returned by Algorithm~\ref{alg:blackbox1} when the current clique collection is $\mathcal{C}_t$. {Write \(\mathbf{1}_{S_t}\in\mathbb{R}^{|V_H|}\) for the characteristic vector of \(S_t\), with \((\mathbf{1}_{S_t})_v=1\) for \(v\in S_t\) and \((\mathbf{1}_{S_t})_v=0\) otherwise.}
By Proposition~\ref{prop:blackbox1-correctness}, $S_t$ intersects every clique in $\mathcal{C}_t$. Thus, we have
\(
\mathbf{1}_{S_t}\in L^t.
\)
If the algorithm proceeds to the next inner iteration, then it has found a new maximum clique $C_{t+1}$ in $H\setminus S_t$. Since $S_t$ and $C_{t+1}$ are disjoint, we have \(
\mathbf{1}_{S_t}\notin L^{t+1}.
\) This indicates that $L^{t+1}\subsetneq L^t.$ Thus, the dimension of the affine subspace decreases by at least one at each inner iteration. We know that the affine subspaces are in $\mathbb{R}^{|V_H|}$, hence the inner loop can iterate at most $|V_H|\leq |V|$ times. 
\end{proof}



It follows from the analysis above that if the clique decision oracle runs in polynomial time, then Algorithm~\ref{alg:find-critical-stable-set} runs in polynomial time.


\section{Certifying nonexistence of a clique of a given size in a perfect graph}\label{sec:certificate-for-refuting-clique-of-size-k}
In the previous section, we showed that optimal coloring of perfect graphs can be reduced, in polynomial time and via combinatorial operations, to deciding whether the graph has a clique of a given size. In this section, we provide a certificate for the ``no'' answer to this question in terms of counting closed walks in the graph with certain weights assigned to its edges and nonedges. This certificate will play a central role in the design of our main algorithm in the next section.



We begin by introducing notation for walks in weighted graphs. If $G$ is a weighted graph, we write \(A_G\in\symm\) for its weighted adjacency matrix, where \((A_G)_{ij}\) denotes the weight of the edge \(ij\). The diagonal entries of $A_G$ denote the weight of self loops. We now define the total weight of walks in such a graph.

\begin{definition}[Total weight of walks]\label{def:total-weight-of-walks}
Let $T$ be a positive integer and $G$ be a weighted graph with weighted adjacency matrix $A_G \in \symm$. We define $\psi_{T,A_G}(i,j)$ to be the total weight of length $T$ walks in $G$ from vertex $i$ to vertex $j$:
   \[
   \psi_{T,A_G}(i,j) = \sum_{(v_1,\dots, v_{T+1}): v_1 = i, v_{T+1}=j} \prod_{m=1}^{T} (A_G){v_m, v_{m+1}}.
   \]
   Here the sum is over all $(T+1)$-tuples $(v_1, v_2,\ldots, v_{T+1})$ of vertices with $v_1 = i$ and $v_{T+1} = j$. We denote the total weight of (closed) walks of length $T$ from vertex $i$ to itself by
    \(
    \psi_{T,A_G}(i) \coloneqq \psi_{T,A_G}(i,i).
    \)
\end{definition}

In our algorithm, we will assign weights to edges and nonedges of a perfect graph $G(V,E)$ according to a probability vector $p\in\Delta_{|\bar{E}|+1}$ to get a weighted graph $G_p$. More precisely, the weighted adjacency matrix is given by
\begin{equation}\label{eq:AGp}
    A_{G_p} = \sum_{ij\in \bar{E}} -p_{ij} Y_{ij} + p_0 J + I,
\end{equation}
where the matrices $Y_{ij}$, $J$, $I$ are as introduced in Section~\ref{sec:notations}, and $p = (p_0, p_{ij})$ for $ij \in \bar{E}$ is a probability vector defined over the nonedges of $G$ and a special index $0$. Figure~\ref{fig:first-step-example} illustrates an example of a perfect graph $G$ and the resulting weighted graph $G_p$ when $p$ is the uniform distribution. 



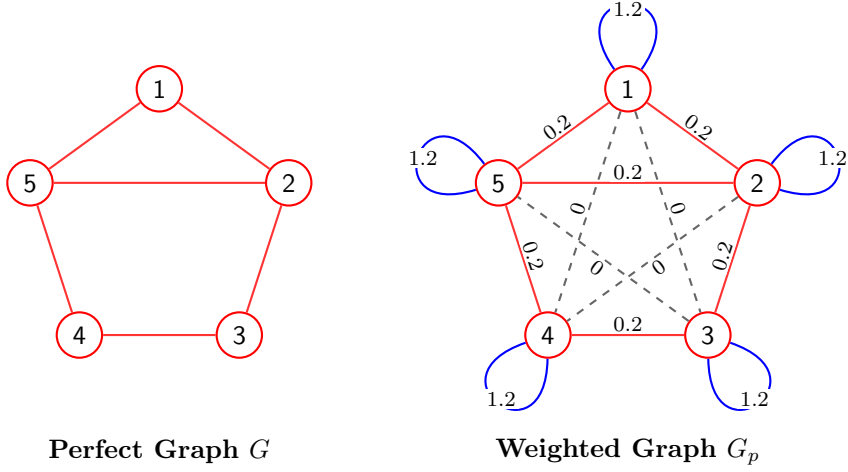
\begin{figure}[htbp]
    \centering
    \begin{tikzpicture}[
        node style/.style={
            draw=red, circle, thick, fill=white,
            minimum size=6mm, inner sep=0pt, font=\sffamily\small
        },
        red edge/.style={draw=red!80, thick},
        zero edge/.style={draw=black!60, dashed, thick},
        loop style/.style={draw=blue, thick, looseness=10},
        weight style/.style={
            font=\scriptsize\sffamily,
            fill=white,
            inner sep=1pt
        },
        loop weight/.style={
            font=\scriptsize\sffamily,
            fill=white,
            inner sep=1pt
        },
        subtitle/.style={font=\bfseries\small}
    ]

        \begin{scope}[shift={(0,0)}]
            \node[node style] (1) at (90:1.8) {1};
            \node[node style] (2) at (18:1.8) {2};
            \node[node style] (3) at (-54:1.8) {3};
            \node[node style] (4) at (-126:1.8) {4};
            \node[node style] (5) at (162:1.8) {5};

            \draw[red edge] (1) -- (2) -- (3) -- (4) -- (5) -- (1);
            \draw[red edge] (2) -- (5);

            \node[subtitle] at (0,-3) {Perfect Graph $G$};
        \end{scope}

        \begin{scope}[shift={(6.2,0)}]
            \foreach \i [evaluate=\i as \ang using {90-72*(\i-1)}] in {1,...,5} {
                \node[node style] (n\i) at (\ang:1.8) {\i};

                \pgfmathsetmacro{\outang}{\ang + 35}
                \pgfmathsetmacro{\inang}{\ang - 35}
                \draw[loop style] (n\i) to [out=\outang, in=\inang] (n\i);

                \node[loop weight] at (\ang:2.85) {$1.2$};
            }

            \draw[red edge] (n1) -- node[weight style, midway, sloped, above] {$0.2$} (n2);
            \draw[red edge] (n2) -- node[weight style, midway, sloped, above] {$0.2$} (n3);
            \draw[red edge] (n3) -- node[weight style, midway, sloped, above] {$0.2$} (n4);
            \draw[red edge] (n4) -- node[weight style, midway, sloped, above] {$0.2$} (n5);
            \draw[red edge] (n5) -- node[weight style, midway, sloped, above] {$0.2$} (n1);
            \draw[red edge] (n2) -- node[weight style, midway, above] {$0.2$} (n5);

            \draw[zero edge] (n1) -- node[weight style, midway, sloped, above] {$0$} (n3);
            \draw[zero edge] (n1) -- node[weight style, midway, sloped, above] {$0$} (n4);
            \draw[zero edge] (n2) -- node[weight style, midway, sloped, below] {$0$} (n4);
            \draw[zero edge] (n3) -- node[weight style, midway, sloped, below] {$0$} (n5);

            \node[subtitle] at (0,-3) {Weighted Graph $G_p$};
        \end{scope}

    \end{tikzpicture}
    \caption{An example of a perfect graph $G$ and its associated weighted graph $G_p$ when $p$ is the uniform distribution.}
    \label{fig:first-step-example}
\end{figure}

The following theorem provides a certificate of nonexistence of a clique of a given size in a perfect graph $G$ by counting the total weight of closed walks of a certain length in $G_p$ for some probability vector $p$. 
\begin{theorem}[Certificate of nonexistence of a clique of size $k$]\label{thm:inexact-certificate-for-refuting-clique-of-size-k}
    Let $G(V,E)$ be a perfect graph with $|V| = n$, $k\in \{2,\dots, n\}$, and fix $\hat{\epsilon}\leq  \frac{1}{(k+1)\pigl((k-1)n(n-1) + 2\pigr)}$. Then for any $\hat{T} \geq \frac{1}{\hat{\epsilon}} \ln (\frac{2n}{\hat{\epsilon}})  $, the following two statements are equivalent:
    \begin{enumerate}[label=\roman*.]
        \renewcommand\theenumi{\roman{enumi}}      
        \renewcommand\labelenumi{(\theenumi)}      
        \item $G$ has no clique of size $k$,
        
        \item
        there exists a probability vector $p\in\Delta_{|\bar{E}| + 1}$, with $p = (p_0, p_{ij})$ for $ij\in \bar{E}$, such that the walk count ratio of the weighted graph with adjacency matrix $A_{G_p} = \sum_{ij\in \bar{E}}-p_{ij}Y_{ij} + p_0 J + I$ satisfies
    $$\max_{i\in S}\frac{\psi_{2\hat{T}+1,A_{G_p}}(i)}{\psi_{2\hat{T},A_{G_p}}(i)}< (1-\hat{\epsilon})(p_0k+1),$$
    where 
        \[
S = \left\{ i \in [n] \mid \psi_{2\hat T,A_{G_p}}(i) > 0 \right\}.
\]
    
    \end{enumerate}
    
\end{theorem}

Throughout the paper, we refer to the quantity 
$\max_{i\in S}\frac{\psi_{2\hat{T}+1,A_{G_p}}(i)}{\psi_{2\hat{T},A_{G_p}}(i)}$
as the \emph{walk count ratio}. This quantity depends on the graph $G$, the probability vector $p$, and the integer $\hat{T}$. We defer the proof of Theorem~\ref{thm:inexact-certificate-for-refuting-clique-of-size-k} until after some auxiliary lemmas. 
The first lemma recalls the connection between walk counting and matrix multiplication. We omit the proof of this standard fact.
\begin{lemma}\label{lem:connection-between-walk-counting-and-matrix-multiplication}
    Let $G$ be a weighted graph on $n$ vertices with weighted adjacency matrix $A_G$. Then for any $i,j\in [n]$ and any integer $t\geq 1$, we have $$\psi_{t,A_G}(i,j) = (A_G^t)_{ij}.$$
\end{lemma}



The next lemma shows that when the weighted adjacency matrix is nonzero and positive semidefinite, its largest eigenvalue is well approximated by the walk count ratio for sufficiently large \(\hat T\). The proof of this lemma appears in Appendix~\ref{app:proof:approximation-via-counting-paths}.

\begin{lemma}\label{lem:approximation-via-counting-paths}
    Let $G$ be a weighted graph with a nonzero positive semidefinite weighted adjacency matrix $A_G$, and let $\psi$ be defined as in Definition~\ref{def:total-weight-of-walks}. Then for any $\epsilon\in (0,1]$ and $\hat{T} \geq \frac{1}{\epsilon}\ln(\frac{2n}{\epsilon})$,
    we have
    $$(1-\epsilon)\lambda_\max(A_G) \leq \max_{i\in S}\frac{\psi_{2\hat{T}+1,A_G}(i)}{\psi_{2\hat{T},A_G}(i)} \leq \lambda_\max(A_G),$$
    where 
    \[
S = \left\{ i \in [n] \mid \psi_{2\hat T,A_G}(i) > 0 \right\}.
\]

\end{lemma}


Our final lemma shows that the matrix $A_{G_p}$ defined in~\eqref{eq:AGp} always satisfies the assumption of Lemma~\ref{lem:approximation-via-counting-paths}. The proof of this lemma can be found in Appendix~\ref{app:proof:lower-bound-eigenvalue}.

\begin{lemma}\label{lem:lower-bound-eigenvalue}
    For any graph $G(V,E)$, let $A_{G_p} = \sum_{ij\in \bar{E}}-p_{ij} Y_{ij} + p_0 J + I$, where $p = (p_0, p_{ij})$ for $ij\in \bar{E}$ is a probability vector. Then $A_{G_p}$ is nonzero and positive semidefinite.
\end{lemma}

We are now ready to prove Theorem~\ref{thm:inexact-certificate-for-refuting-clique-of-size-k}.
\begin{proof}[Proof of Theorem~\ref{thm:inexact-certificate-for-refuting-clique-of-size-k}]
    Fix $k\in \{2,\dots,n\}$, $\hat{\epsilon}\leq  \frac{1}{(k+1)\mathlarger((k-1)n(n-1) + 2\mathlarger)}$, and $\hat{T} \geq \frac{1}{\hat{\epsilon}} \ln (\frac{2n}{\hat{\epsilon}})  $. 

    \textbf{(ii) $\Rightarrow$ (i).} Suppose there is a probability vector $p = (p_0, p_{ij})$ for $ij\in \bar{E}$, such that
    $$\max_{i\in S}\frac{\psi_{2\hat{T}+1,A_{G_p}}(i)}{\psi_{2\hat{T},A_{G_p}}(i)}< (1-\hat{\epsilon})(p_0k+1).$$

    By Lemma~\ref{lem:lower-bound-eigenvalue}, we know $A_{G_p}$ is nonzero and positive semidefinite. Hence Lemma~\ref{lem:approximation-via-counting-paths} applies and gives
    $$(1-\hat{\epsilon})\lambda_\max(A_{G_p}) \leq \max_{i\in S}\frac{\psi_{2\hat{T}+1,A_{G_p}}(i)}{\psi_{2\hat{T},A_{G_p}}(i)} < (1-\hat{\epsilon})(p_0k + 1).$$
    Since $1-\hat{\epsilon}>0$, it follows that
    $$\lambda_\max\left(\sum_{ij\in \bar{E}}-p_{ij}Y_{ij} + p_0 J\right) = \lambda_\max(A_{G_p})-1 < p_0k.$$


    Suppose for the sake of contradiction that \(G\) contains a clique \(C\) of size \(k\). Consider the $|C|\times|C|$ matrix
\[
\left(\sum_{ij\in\bar E}-p_{ij}Y_{ij}+p_0J\right)_C,
\]
which is the principal submatrix of $\sum_{ij\in \bar{E}}-p_{ij}Y_{ij} + p_0 J$ indexed by the vertices in $C$. Since \(C\) is a clique, no pair of vertices in \(C\) is a nonedge of \(G\). Hence all terms \(-p_{ij}Y_{ij}\) vanish in this principal submatrix, and so
\[
\left(\sum_{ij\in\bar E}-p_{ij}Y_{ij}+p_0J\right)_C
=
p_0J_C.
\]
Thus,
\[
\lambda_{\max}\!\left(\sum_{ij\in\bar E}-p_{ij}Y_{ij}+p_0J\right)
\ge
\lambda_{\max}(p_0J_C)
=
p_0k,
\]
where the inequality follows, for example, from the Rayleigh quotient inequality. This contradicts the previous strict inequality. Therefore, \(G\) has no clique of size \(k\).

    \textbf{(i) $\Rightarrow$ (ii):} 
Suppose $G$ has no clique of size $k$. Since $G$ is perfect, it follows that it is $(k-1)$-colorable. Pick any $(k-1)$-coloring $\mathcal{C}$ of $G$. Let $w_0 = \frac{1}{k-1}$ and $w_{ij} = \mathbf{1}(\mathcal{C}(i) = \mathcal{C}(j))$.  Let 
$$p_0 = \frac{w_0}{\sum_{ij\in \bar{E}} w_{ij} + w_0},\  p_{ij} = \frac{w_{ij}}{\sum_{ij\in \bar{E}} w_{ij} + w_0}.$$ 

Since vertices of the same color are nonadjacent, $p_{ij}$ is nonzero only when $ij\in \bar{E}$. 
We first claim that
\begin{equation}\label{eq:delta}
    \lambda_\max\left(\sum_{ij\in \bar{E}}-p_{ij}Y_{ij} + p_0 J\right) \leq p_0(k-1). 
\end{equation}
This is equivalent to showing that
$$\lambda_\max\left(\sum_{ij\in \bar{E}}-w_{ij}Y_{ij} +  \frac{1}{k-1}J\right) \leq 1,$$
or that for every $x\in \Real{n}$,
\begin{equation}\label{eq:coloring-ineq}
    \frac{1}{k-1}\left(\sum_{i=1}^n x_i\right)^2 \leq \sum_{i=1}^n x_i^2 + \sum_{ij\in \bar{E}}2w_{ij}x_ix_j.
\end{equation}

Let $I_{\mathcal{C},1},I_{\mathcal{C},2},\dots,I_{\mathcal{C},k-1}$ denote the color classes of the coloring $\mathcal{C}$. The right-hand side of~\eqref{eq:coloring-ineq} can be rewritten as
\begin{align*}
    \sum_{i=1}^n x_i^2 + 2\sum_{ij\in \bar{E}}\mathbf{1}(\mathcal{C}(i) = \mathcal{C}(j))x_ix_j &= \sum_{l = 1}^{k-1}\sum_{i\in I_{\mathcal{C},l}} x_i^2 + \sum_{l=1}^{k-1}\sum_{
    \substack{i<j\\ i,j \in I_{\mathcal{C},l}}}2x_ix_j \\
    &= \sum_{l = 1}^{k-1}\left(\sum_{i\in I_{\mathcal{C},l}} x_i\right)^2.
\end{align*}
Applying the Cauchy-Schwarz inequality to the left-hand side of~\eqref{eq:coloring-ineq}, we have
\begin{align*}
    \frac{1}{k-1}\left(\sum_{i = 1}^{n} x_i\right)^2 = \frac{1}{k-1}\left(\sum_{l = 1}^{k-1}\sum_{i\in I_{\mathcal{C},l}} x_i\right)^2 \leq \sum_{l = 1}^{k-1}\left(\sum_{i\in I_{\mathcal{C},l}} x_i\right)^2.
\end{align*}
Hence, the inequality~\eqref{eq:delta} is proven. Using this inequality, we now verify that $p \coloneqq (p_0,p_{ij})$ for $ij\in\bar{E}$ satisfies condition (ii) of the theorem. Since $w_{ij}\in \{0,1\}$, we have
$$p_0 = \frac{\frac{1}{k-1}}{\sum_{ij\in \bar{E}} w_{ij} + \frac{1}{k-1}} \geq \frac{\frac{1}{k-1}}{\frac{n(n-1)}{2} + \frac{1}{k-1}} = \frac{2}{(k-1)n(n-1) + 2}.$$
Recalling that $\hat{\epsilon} \leq \frac{1}{(k+1)\pigl((k-1)n(n-1) + 2\pigr)}$, it follows that
    \begin{align*}
        (p_0k +1)\hat{\epsilon} &\leq \frac{1}{(k-1)n(n-1) + 2} < p_0,
    \end{align*}
    and therefore
    \begin{equation}\label{eq:coloring-algebra-step}
        p_0(k-1) + 1 < (1-\hat{\epsilon})(p_0k + 1).
    \end{equation}
We now observe that
\begin{align*}
\max_{i\in S}
\frac{\psi_{2\hat{T}+1,A_{G_p}}(i)}
     {\psi_{2\hat{T},A_{G_p}}(i)}
&\le
\lambda_{\max}(A_{G_p}) \\
&=
\lambda_{\max}\!\left(\sum_{ij\in\bar E}-p_{ij}Y_{ij}+p_0J\right)+1 \\
&\le
p_0(k-1)+1 \\
&<
(1-\hat{\epsilon})(p_0k+1),
\end{align*}
where the first inequality follows from Lemma~\ref{lem:approximation-via-counting-paths}, the second from~\eqref{eq:delta}, and the third from~\eqref{eq:coloring-algebra-step}.
This completes the proof.\footnote{Note from the proof that the implication \((ii)\Rightarrow(i)\) in Theorem~\ref{thm:inexact-certificate-for-refuting-clique-of-size-k} holds for arbitrary graphs, while the implication \((i)\Rightarrow(ii)\) holds for graphs \(G\) that satisfy
\(
\chi(G)=\omega(G).
\)}
        
\end{proof}


\section{Deciding if a perfect graph has a clique of a given size} \label{sec:alg-analysis}
Following the reduction established in Section~\ref{sec:reduction-color-clique-decision}, the problem of optimally coloring a perfect graph $G$ reduces to deciding if $G$ has a clique of a given size $k \in \{2, \dots, |V|\}$. In this section, we present a combinatorial algorithm (Algorithm~\ref{alg:decision-algorithm}) for this decision problem. Following Theorem~\ref{thm:inexact-certificate-for-refuting-clique-of-size-k}, the algorithm attempts to construct a probability vector that certifies the nonexistence of a $k$-clique. To do so, it iteratively updates a probability vector based on closed walk counts in $G$ with certain weights assigned to its edges and nonedges. We begin by detailing the walk counting procedure (Algorithm~\ref{alg:counting-path-alg}). We then present the main algorithm (Algorithm~\ref{alg:decision-algorithm}) and establish its correctness and polynomial-time complexity via a Lyapunov function argument.

\subsection{The walk counting algorithm}

The walk counting algorithm (Algorithm~\ref{alg:counting-path-alg}) takes as input a weighted adjacency matrix $A\in\Q{n\times n}$ and a positive integer $\hat{T}$. It outputs a scalar $\beta$, which we call the walk count ratio, and a vector $x\in\Q{n}$, which we call the walk count vector. The expression for the walk count ratio is the same as Theorem~\ref{thm:inexact-certificate-for-refuting-clique-of-size-k}:

$$\beta= \max_{i\in S}\frac{\psi_{2\hat{T}+1,A}(i)}{\psi_{2\hat{T},A}(i)},$$
where $\psi$ is as in Definition~\ref{def:total-weight-of-walks} and $S = \{i\in [n]\mid \psi_{2\hat{T},A}(i) > 0\}.$ Let $i_0$ be an index that achieves this maximum. The walk count vector collects the total weight of walks of length $\hat{T}$ from $i_0$ to all the other vertices; \ie it is defined entrywise as
$$x_j = \psi_{\hat{T},A}(i_0,j), \quad \forall j \in [n].$$

\begin{algorithm}
    \SetAlgoLined
    \caption{Walk counting}\label{alg:counting-path-alg}
    \textbf{Input:} A weighted adjacency matrix $A\in \Q{n\times n}$, a positive integer $\hat{T}\in \Z{} $\\
    \textbf{Output:} A scalar $\beta\in\Q{}$ and a vector $x\in\Q{n}$

    \textbf{Initialize:}  $\psi_{1,A}(i,j) \gets A_{ij}$ for all $i,j\in [n]$

    \For{$t\gets 1\  \KwTo\  2\hat{T}$}{

        \For{$i \gets 1\  \KwTo\  n$}{
            \For{$j\gets 1\  \KwTo\  n$}{
                $\psi_{t+1,A}(i, j) \gets \sum_{i_1\in [n]} A_{i_1,j}\cdot \psi_{t,A}(i,i_1)$
            }
        }
    }
    $S \gets \left\{ i \in [n] \mid \psi_{2\hat T,A}(i) > 0 \right\}$ \\
    $\beta \gets \max_{i\in S}\frac{\psi_{2\hat{T}+1,A}(i)}{\psi_{2\hat{T},A}(i)}$\\
    $i_0\gets \arg\max_{i\in S} \frac{\psi_{2\hat{T}+1,A}(i)}{\psi_{2\hat{T},A}(i)}$ \tcp*[f]{pick any one if the maximizer is not unique}%

    \For{$j \gets 1\  \KwTo\  n$}{
            $x_j\gets \psi_{\hat{T},A}(i_0,j)$
    }

\textbf{Output:} $(\beta, x)$

\end{algorithm}

Let us argue that the walk counting algorithm correctly returns the walk count ratio $\beta$ and the walk count vector $x$ by showing that the nested loops in this algorithm correctly compute $\psi_{m,A}(i,j)$, the total weight of walks of length $m$ from vertex $i$ to vertex $j$. For $m=1$, the algorithm initializes $\psi_{1,A}(i,j) = A_{ij}$, which is consistent with the definition. For $m \geq 1$, the total weight of walks satisfies the recurrence:
\begin{align*}
    \psi_{m+1,A}(i,j) &= \sum_{k \in [n]} A_{k,j} \sum_{\substack{\text{length } m \text{ walk } \\ \ell: \ell(1)=i,\, \ell(m+1)=k}} \prod_{r=1}^{m} A_{\ell(r),\ell(r+1)} \\
    &= \sum_{k \in [n]} A_{k,j} \cdot \psi_{m ,A}(i, k).
\end{align*}
This recurrence is precisely the update rule in the innermost loop of Algorithm~\ref{alg:counting-path-alg}. The time complexity of Algorithm~\ref{alg:counting-path-alg} is $O(n^3\hat{T})$.

\subsection{The main algorithm}\label{sec:correctness-of-algorithm}

In this section, we present our main algorithm (Algorithm~\ref{alg:decision-algorithm}), which given a perfect graph $G(V,E)$ and an integer $k \in \{2, \dots, |V|\}$ determines whether $G$ contains a clique of size $k$.\footnote{As the proof of correctness of Algorithm~\ref{alg:decision-algorithm} will demonstrate, this algorithm will output the correct decision more generally on graphs for which the clique number and the chromatic number coincide.}

\begin{algorithm}
    \SetAlgoLined
    \caption{Deciding if a perfect graph $G$ has a clique of size $k$}\label{alg:decision-algorithm}
    \textbf{Input:} Perfect graph $G(V,E)$ with $|V| = n$, integer $k\in\{2,\dots,n\}$

    {\bf Output:} ``YES'' or ``NO''
    
    Set parameters: $\epsilon = \frac{1}{2n}\frac{1}{n(n-1)(k-1) + 2}$, $\hat{\epsilon} \in \left(0,  \frac{1}{(k+1)\big((k-1)n(n-1) + 2\big)}\right]$, integer~$T > 4\ln(|\bar{E}|+1)n^2\Big((k-1)n(n-1)+2\Big)^2$, integer $ \hat{T}\geq \frac{1}{\hat{\epsilon}} \ln (\frac{2n}{\hat{\epsilon}})$ \tcp*[f]{any value of $\hat{\epsilon}, T, \hat{T}$ which results in $T$ and $\hat{T}$ having polynomial magnitude is acceptable}%

    Initialize weights:
     $w_{ij}^1 \gets 1, \ \text{for } ij\in \bar{E}$, $w_0^1 \gets 1$\\

    \For{$t\gets 1\  \KwTo\  T$}{
    Normalize the weights:
        $$p^t_{ij} \gets \frac{w^t_{ij}}{\sum_{ij\in \bar{E}} w^t_{ij} + w_0^t},\quad p^t_0 \gets \frac{w_0^t}{\sum_{ij\in \bar{E}} w^t_{ij} + w_0^t}$$
         
        Construct a weighted graph $G_{p^t}$ with weighted adjacency matrix $A_{G_{p^t}} = \sum_{ij\in \bar{E}}-p^t_{ij}Y_{ij} + p_0^t J + I$.

        $(\beta^t, x^t)\gets $ Algorithm~\ref{alg:counting-path-alg}$(A_{G_{p^t}}, \hat{T})$ \tcp*[f]{compute walk count ratio and walk counts}%

      \If{$\beta^t < (1-\hat{\epsilon})( p^t_0k + 1)$}{
            \textbf{return} ``NO''
        }
      
        Update weights based on walk counts:
        $$w^{t+1}_{ij} \gets w^t_{ij}\,\left(1 + \frac{\epsilon}{n} \frac{2x^t_ix^t_j}{\|x^t\|_2^2}\right),
         \ w^{t+1}_0 \gets w^t_0\,\left(1 - \frac{\epsilon}{n}\left(\frac{(\sum_{i\in[n]} x_i^t)^2}{\|x^t\|_2^2}-k\right)\right)$$
    }
    \textbf{Return} ``YES'' 
\end{algorithm}

In each iteration, the algorithm maintains a probability vector $p^t\in\Delta_{|\bar{E}|+1}$, which defines a weighted graph $G_{p_t}$ with adjacency matrix $A_{G_{p_t}}$ as in~\eqref{eq:AGp}.
The algorithm then invokes the walk counting algorithm (Algorithm~\ref{alg:counting-path-alg}) on $G_{p^t}$ to compute the walk count ratio $\beta^t$ and the walk count vector $x^t$. 

If $\beta^t < (1-\hat{\epsilon})(p_0^tk + 1)$, Theorem~\ref{thm:inexact-certificate-for-refuting-clique-of-size-k} guarantees that there is no clique of size $k$ in $G$; hence, the algorithm returns ``NO''.  Otherwise, the weights are updated based on the walk counts in $x^t$, and the process repeats. A combinatorial interpretation of the quantities appearing in the weight update step is as follows. Let $i_0^t$ denote the index maximizing the walk count ratio $\beta^t$. Recall that $x^t_j$ represents the total weight of walks of length $\hat{T}$ from $i_0^t$ to $j$ in $G_{p^t}$. The quantities in the update rule have the following combinatorial interpretation:
\begin{itemize}
    \item $x_i^tx_j^t$ is the total weight of walks of length $2\hat{T}$ from $i$ to $j$ passing through $i_0^t$ at the midpoint.
    \item $\|x^t\|_2^2$ is the total weight of closed walks of length $2\hat{T}$ passing through $i_0^t$ at the midpoint.
    \item $(\sum_{i\in[n]} x_i^t)^2$ is the total weight of all walks of length $2\hat{T}$ passing through $i_0^t$ at the midpoint.
\end{itemize}

If the loop terminates after $T$ iterations without a ``NO'' certificate, the algorithm returns ``YES''. We defer the proof of correctness of Algorithm~\ref{alg:decision-algorithm} to the next section. We note that since the bounds on $T$ and $\hat{T}$ are of polynomial magnitude, Algorithm~\ref{alg:decision-algorithm} runs in polynomial time.

%

\subsection{Proof of correctness of the main algorithm}\label{sec:proof_of_correctness}

In this section, we prove the correctness of Algorithm~\ref{alg:decision-algorithm}. The strategy is as follows: If there is a clique of size $k$ in the perfect graph $G$, 
Theorem~\ref{thm:inexact-certificate-for-refuting-clique-of-size-k} implies that no iterate $p^t$ can satisfy the termination condition (in line 9), so the algorithm will necessarily return `YES’. If there is no clique of size $k$ in $G$, we define a Lyapunov function $L:\Delta_{|\bar{E}|+1} \to\Real{}$ to measure how far the current probability vector $p^t$ is from the target probability vector constructed in the proof of Theorem~\ref{thm:inexact-certificate-for-refuting-clique-of-size-k}. We further argue that this Lyapunov function is nonnegative, not too large at the first iterate, and decreases at least by a certain fixed amount in each iteration of the algorithm. From this analysis, it follows that the probability vectors produced by the algorithm must necessarily satisfy the termination condition within $T$ steps, leading to the output `NO' in polynomial time.





Let us now formalize the proof strategy. Suppose the perfect graph $G$ has no clique of size $k$. By perfectness, we can pick a $(k-1)$-coloring of $G$, denoted by $\mathcal{C}$. Then, we let $w_0^\star = \frac{1}{k-1},
w_{ij}^\star =~\mathbf{1}(\mathcal{C}(i) =~\mathcal{C}(j))$, and 
$$p_0^\star = \frac{w_0^\star}{\sum_{ij\in \bar{E}} w_{ij}^\star + w_0^\star},\  p_{ij}^\star = \frac{w_{ij}^\star}{\sum_{ij\in \bar{E}} w_{ij}^\star + w_0^\star}.$$
And we write $p^\star = (p^\star_0, p_{ij}^\star)$, for $ij\in\bar E$. From the proof of Theorem~\ref{thm:inexact-certificate-for-refuting-clique-of-size-k} (more specifically inequality~\eqref{eq:delta}), we know that ${p^\star}$ satisfies
\begin{equation}\label{eq:p-star-bound}
    \lambda_\max\left(\sum_{ij\in \bar{E}} - {p}^\star_{ij} Y_{ij} + p^\star_0 (J - kI)\right) \leq -p_0^\star\leq -\frac{2}{(k-1)n(n-1) + 2}.
\end{equation}
Consider the following Lyapunov function $L:\Delta_{|\bar{E}|+1} \to\Real{}$ defined as
\begin{equation}\label{eq:def-lyapunov}
    L(p) = D_{KL}(p^\star \| p) \coloneqq \sum_{ij\in \bar{E}} p^\star_{ij} \ln \frac{p^\star_{ij}}{p_{ij}} + p^\star_{0}\ln \frac{p^\star_{0}}{p_{0}}.
\end{equation}

{The choice of the Kullback-Leibler (KL) divergence as a Lyapunov function has been made previously; see, e.g.,~\cite{freund1997decision,arora2012mwu}.} Let $x^t\in\Real{n}$ be the walk count vector at iteration $t$ of Algorithm~\ref{alg:decision-algorithm}, and let
$$r^t_{ij} = - \frac{2x^t_ix^t_j}{\|x^t\|_2^2},\ \text{for } ij\in\bar{E},\ \ r^t_0 = \frac{(\sum_{i\in[n]} x_i^t)^2}{\|x^t\|_2^2}-k.$$
Let $p_0^t,p_{ij}^t\in[0,1]$ be the scalars at iteration $t$ of Algorithm~\ref{alg:decision-algorithm}. We define the vectors $p^t\in\Delta_{|\bar E|+1}$ and $r^t\in\Real{|\bar E|+1}$ as
\begin{equation}\label{eq:def-p-t-r-t}
\begin{aligned}
        p^t = (p^t_0,p^t_{ij}), \\
    r^t =(r^t_0, r^t_{ij}).
\end{aligned}
\end{equation}  To prove that our Lyapunov function decreases along the iterations of Algorithm~\ref{alg:decision-algorithm}, we first establish two lemmas.

\begin{lemma}\label{lem:rho-bound}
    Let $r^t$ be the vector defined in~\eqref{eq:def-p-t-r-t} at iteration $t$ of Algorithm~\ref{alg:decision-algorithm} with input $G(V,E)$ and $k\in \{2,\dots,n\}$. We have $\|r^t\|_\infty \leq n$.
\end{lemma}

\begin{proof}


Note that for $ij\in \bar{E}$, we have $$ |r_{ij}^t| = \left|\frac{2x^t_ix^t_j}{\|x^t\|_2^2}\right|\leq\frac{(x_i^t)^2 + (x_j^t)^2}{\|x^t\|_2^2} \leq 1.$$ 
Furthermore,
\[
r^t_0 = \frac{(\sum_{i\in[n]} x_i^t)^2}{\|x^t\|_2^2}-k \leq \frac{\|x^t\|_1^2}{\|x^t\|_2^2} - k  \leq n - k,
\]
and 
\[
r^t_0 = \frac{(\sum_{i\in[n]} x_i^t)^2}{\|x^t\|_2^2}-k \geq -k,
\]
and hence $|r_0^t| \leq n$. Therefore, $\|r^t\|_\infty \leq n$.
\end{proof}

\begin{lemma}\label{lem:lyapunov-key-inequality}Let $p^t$ and $r^t$ be the vectors defined in~\eqref{eq:def-p-t-r-t} at iteration $t$ of Algorithm~\ref{alg:decision-algorithm} with input $G(V,E)$ and $k\in \{2,\dots,n\}$. We have
$$\frac{p_{m}^t}{p_{m}^{t+1}} \leq \frac{\exp(-\frac{\epsilon}{n}\langle p^t,r^t\rangle)}{1-\frac{\epsilon}{n}r^t_{m}},$$
for all $m\in \bar{E}\cup \{0\}$.
\end{lemma}

\begin{proof}
Define $\Phi^t = \sum_{ij\in \bar{E}} w^t_{ij} + w^t_0$. From the update rule of $w^t$ and Lemma~\ref{lem:rho-bound}, we know that $\Phi^t > 0$.
Observe that
    \begin{align*}
        \Phi^{t+1} &= \sum_{k\in \bar{E}\cup \{0\}} w_k^{t+1} \\
    &= \sum_{k\in \bar{E}\cup \{0\}} w_k^t \,(1 - \frac{\epsilon}{n}r^t_k)\\
    &= \Phi^t - \frac{\epsilon}{n}\Phi^t\ip{p^t}{r^t}
    \\
    &= \Phi^t(1-\frac{\epsilon}{n}\ip{p^t}{r^t})\\
    &\le \Phi^t \exp \Bigl(-\tfrac{\epsilon}{n}\ip{p^t}{r^t}\Bigr),
    \end{align*}
    where the last step follows from the facts that $1+x\le e^x$ for all $x\in\Real{}$ and $\Phi_t > 0$.
    Since from Lemma~\ref{lem:rho-bound} we know that $1 - \frac{\epsilon}{n}r^t_m > 0$ for all $m\in \bar{E}\cup \{0\}$, we have
\begin{align*}
    \frac{p_{m}^t}{p_{m}^{t+1}} &= \frac{w^t_m}{w^{t+1}_m}\cdot \frac{\Phi^{t+1}}{\Phi^t} = \frac{1}{1-\frac{\epsilon}{n }r^t_m}\cdot \frac{\Phi^{t+1}}{\Phi^t} \leq \frac{\exp(-\frac{\epsilon}{n }\langle p^t,r^t \rangle)}{1-\frac{\epsilon}{n}r^t_m}.
\end{align*}

\end{proof}

We can now bound the amount by which the proposed Lyapunov function decreases at every iteration of Algorithm~\ref{alg:decision-algorithm}.


\begin{theorem}[Lyapunov function decrease]\label{thm:lyapunov-decrease}
Let $G(V,E)$ be a perfect graph and $k\in \{2,\dots,n\}$. Let $p^t$ be the vector defined in~\eqref{eq:def-p-t-r-t} at iteration $t$ of Algorithm~\ref{alg:decision-algorithm} with input $G(V,E)$ and $k$. If $G$ has no clique of size $k$, then either the termination condition in Algorithm~\ref{alg:decision-algorithm} is satisfied at iteration $t$, namely,
    $$\max_{i\in S^t} \frac{\psi_{2\hat{T} + 1,A_{G_{p^t}}}(i)}{\psi_{2\hat{T},A_{G_{p^t}}}(i)} < (1-\hat{\epsilon})(p_0^t k + 1),$$
    where $S^t\coloneqq \{i\in [n]\mid \psi_{2\hat T,A_{G_{p^t}}}(i) > 0\}$, or the Lyapunov function $L$ defined in~\eqref{eq:def-lyapunov} satisfies
    \begin{equation}\label{eq:lyapunov-decrease-bound}
        L(p^{t+1}) - L(p^t) \leq -\frac{1}{4n^2((k-1)n(n-1) + 2)^2}.
    \end{equation}
    
\end{theorem}

\begin{proof}
    Fix an iteration $t$, and suppose that Algorithm~\ref{alg:decision-algorithm} does not terminate with output ``NO'' at this iteration. We show that inequality~\eqref{eq:lyapunov-decrease-bound} must hold. Recall that in Algorithm~\ref{alg:decision-algorithm}, we choose
    $$\epsilon = \frac{1}{2n}\frac{1}{n(n-1)(k-1) + 2}\leq \frac{1}{2}.$$
    Together with Lemma~\ref{lem:rho-bound}, this implies that $\frac{\epsilon}{n}r_m^t \leq \frac{1}{2}$ for every index $m\in \bar E \cup \{0\}$. Since the vector $w^1$ of initial weights is at the all-ones vector, it follows that all coordinates of $p^t$ and $p^{t+1}$ are positive, so all logarithms below are well-defined. 
    Using the definition of $L$ in~\eqref{eq:def-lyapunov}, we have
    \begin{align*}
    L(p^{t+1}) - L(p^t) &= \sum_{ij\in \bar{E}}p^\star_{ij}\ln\frac{p^\star_{ij}}{p_{ij}^{t+1}} + p^\star_0\ln\frac{p^\star_0}{p_0^{t+1}} - \sum_{ij\in \bar{E}}p^\star_{ij}\ln\frac{p^\star_{ij}}{p_{ij}^{t}} - p^\star_0\ln\frac{p^\star_0}{p_0^{t}} \\
    &= \sum_{ij\in \bar{E}}p^\star_{ij}\ln\frac{p_{ij}^{t}}{p_{ij}^{t+1}} + p^\star_0\ln\frac{p_0^{t}}{p_0^{t+1}} \\
    &\leq \sum_{ij\in \bar{E}}p^\star_{ij}\left(-\frac{\epsilon}{n}\langle p^t,r^t \rangle - \ln(1-\frac{\epsilon}{n}r^t_{ij})\right) + p^\star_0\left(-\frac{\epsilon}{n}\langle p^t,r^t \rangle - \ln(1-\frac{\epsilon}{n}r^t_0)\right) \\
    &= -\frac{\epsilon}{n}\langle p^t ,r^t \rangle - \sum_{ij\in \bar{E}}p^\star_{ij}\ln(1-\frac{\epsilon}{n}r^t_{ij}) - p^\star_0\ln(1-\frac{\epsilon}{n}r^t_0)\\
    &\leq -\frac{\epsilon}{n}\langle p^t - p^\star,r^t \rangle + \sum_{ij\in \bar{E}}p^\star_{ij}\left(\frac{\epsilon}{n}\right)^2 (r^t_{ij})^2+ p^\star_0 \left(\frac{\epsilon}{n}\right)^2 (r^t_0)^2\\
    &\leq -\frac{\epsilon}{n}\langle p^t - p^\star,r^t \rangle + \epsilon^2,
\end{align*}
where the first inequality follows from Lemma~\ref{lem:lyapunov-key-inequality}, the second inequality from the fact that for $x\leq 1/2$, $-\ln(1-x)\leq x+x^2$, and the last one from Lemma~\ref{lem:rho-bound}. We next bound $\langle p^t, r^t \rangle$ and $\langle p^\star, r^t \rangle$ separately. Recall that at iteration $t$ of Algorithm~\ref{alg:decision-algorithm}, Algorithm~\ref{alg:counting-path-alg} returns the walk count ratio $\beta^t$ and the walk count vector $x^t$, which in view of Lemma~\ref{lem:connection-between-walk-counting-and-matrix-multiplication} and their definitions satisfy
\begin{align*}
    \beta^t &= \max_{i\in S^t}\frac{\psi_{2\hat{T}+1,A_{G_{p^t}}}(i)}{\psi_{2\hat{T},A_{G_{p^t}}}(i)} \\
    &= \max_{i\in S^t}\frac{(A_{G_{p^t}}^{\hat{T}}e_i)\tran A_{G_{p^t}} (A_{G_{p^t}}^{\hat{T}}e_i)}{\|A_{G_{p^t}}^{\hat{T}}e_i\|_2^2}\\
&=\frac{(x^t)\tran (\sum_{ij\in \bar{E}} -p^t_{ij} Y_{ij} + p^t_0 J + I) x^t}{\|x^t\|_2^2},
\end{align*}
where $S^t = \left\{ i \in [n] \mid \psi_{2\hat T,A_{G_{p^t}}}(i) > 0 \right\}$. Then, by the definitions of $p^t$ and $r^t$ in~\eqref{eq:def-p-t-r-t}, we have
\begin{align*}
    \langle p^t, r^t \rangle &= \frac{(x^t)\tran (\sum_{ij\in \bar{E}} -p^t_{ij} Y_{ij} + p^t_0 (J - kI)) x^t}{\|x^t\|_2^2} =  \beta^t - p_0^tk - 1.
\end{align*}
Since the termination condition of Algorithm~\ref{alg:decision-algorithm} is not satisfied, this gives
\[ \langle p^t, r^t \rangle =  \beta^t - p_0^tk - 1 \geq -\hat{\epsilon} (p_0^t k+1).\]
By the definitions of $p^\star$ and $ r^t$, and by the Rayleigh quotient inequality, we have
\begin{align*}
    \langle p^\star, r^t \rangle &= \frac{(x^t)\tran(\sum_{ij\in \bar{E}} -p^\star_{ij} Y_{ij} + p^\star_0 (J - kI)) x^t}{ \|x^t\|_2^2}\leq \lambda_\max(\sum_{ij\in \bar{E}} -p^\star_{ij} Y_{ij} + p^\star_0 (J - kI)).
\end{align*}
Using the two inequalities above, we obtain
\begin{align*}
    L(p^{t+1}) - L(p^t) &\leq -\frac{\epsilon}{n}(\langle p^t, r^t \rangle - \langle p^\star, r^t \rangle) + \epsilon^2\\
    &\leq -\frac{\epsilon}{n}(-\hat{\epsilon}(p_0^t k + 1) - \lambda_\max(\sum_{ij\in \bar{E}} -p^\star_{ij} Y_{ij} + p^\star_0 (J - kI))) +\epsilon^2\\
    &\leq -\frac{\epsilon}{n}(-\hat{\epsilon}(k+1) + \frac{2}{(k-1)n(n-1) + 2}) +\epsilon^2\\
    &\leq -\frac{\epsilon}{n}\frac{1}{(k-1)n(n-1) + 2} +\epsilon^2\\
    &= -\frac{1}{4n^2((k-1)n(n-1) + 2)^2},
\end{align*}
where the third inequality follows from~\eqref{eq:p-star-bound} and the fact that $p_0^t \leq 1$, the fourth inequality from the fact that $\hat{\epsilon} \leq \frac{1}{(k+1)\pigl((k-1)n(n-1) + 2\pigr)}$, and the equation from our choice of $\epsilon = \frac{1}{2n}\frac{1}{n(n-1)(k-1) + 2}$.
\end{proof}


We are now ready to formalize the proof of correctness of Algorithm~\ref{alg:decision-algorithm}.
\begin{theorem}[Correctness of Algorithm~\ref{alg:decision-algorithm}]
    Given a perfect graph $G(V,E)$ with $|V| = n$ and an integer $k\in\{2,\dots, n\}$, Algorithm~\ref{alg:decision-algorithm} correctly decides if there is a clique of size $k$ in $G$.
\end{theorem}

\begin{proof}
    We split the proof into two cases.
    
    \textbf{Case 1:} $G$ has no clique of size $k$. First recall that $p^1\in\Delta_{|\bar E|+1}$ is the uniform distribution. Thus, we can bound the initial value of the Lyapunov function defined in~\eqref{eq:def-lyapunov} as
$$L(p^1) = \ln(|\bar{E}|+1) + \sum_{ij\in \bar{E}} p^\star_{ij} \ln(p^\star_{ij}) + p_0^\star \ln(p_0^\star)\leq \ln(|\bar{E}|+1).$$
    In each iteration $t$ of Algorithm~\ref{alg:decision-algorithm}, either the termination condition is satisfied and the algorithm correctly outputs ``NO'', or, by Theorem~\ref{thm:lyapunov-decrease}, the Lyapunov function $L$ satisfies
    $$L(p^{t+1}) - L(p^t) \leq -\frac{1}{4n^2((k-1)n(n-1) + 2)^2}.$$
    Suppose for the sake of contradiction that Algorithm~\ref{alg:decision-algorithm} does not return ``NO'' in the first $T$ iterations. Then, 
    \begin{equation}\label{eq:lyapunov-T-1}
        L(p^{T+1}) \leq L(p^1) - \frac{T}{4n^2((k-1)n(n-1)+2)^2}. 
    \end{equation}
    By our choice of \(T\) we have
    \[ \frac{T}{4n^2((k-1)n(n-1)+2)^2}>\ln(|\bar E|+1). \]
    Recalling that $L(p^1)\leq \ln(|\bar E|+1)$, we observe that~\eqref{eq:lyapunov-T-1} implies $L(p^{T+1})<0$. This contradicts the fact that the KL divergence is nonnegative. Therefore, Algorithm~\ref{alg:decision-algorithm} returns ``NO'' within the first $T$ iterations.

    \textbf{Case 2:} There is a clique of size $k$ in $G$.
    Suppose for the sake of contradiction that Algorithm~\ref{alg:decision-algorithm} returns ``NO'' at iteration $t$. This can only happen if the termination condition is satisfied; i.e.,
    $$\beta^t < (1-\hat{\epsilon})(p_0^t k + 1).$$
    By Theorem~\ref{thm:inexact-certificate-for-refuting-clique-of-size-k}, this implies that there is no clique of size $k$ in $G$, contradicting the assumption. Thus, Algorithm~\ref{alg:decision-algorithm} must necessarily return ``YES''.
\end{proof}

\section{Connections with other optimization algorithms}\label{sec:connection}


In a seminal paper~\cite{lovasz1979shannon}, \lovasz introduced the \emph{theta number} $\vartheta(G)$ of a graph $G(V,E)$ as the optimal value of the following semidefinite program (SDP):
    \begin{equation*}\label{opt:lovasz-sdp}
        \begin{aligned}
\vartheta(G) \coloneqq \max_{X \in \symm}\quad & \tr(JX) \\
\text{subject to}\quad 
    & \tr(X) = 1,\\
    & X \succeq 0,\\
    & X_{ij} = 0,\quad \text{for }\,ij\in E.
\end{aligned}
    \end{equation*}
He further showed that for any graph $G$, the theta number of the complement graph $\bar{G}$ satisfies the following inequalities:
$$\omega(G)\leq \vartheta(\bar{G})\leq \chi(G).$$
It follows that for perfect graphs, the above three parameters coincide. Therefore, \emph{any} algorithm that computes the clique number of a perfect graph is, in effect, also solving the above semidefinite program on the complement graph. {In this section, we argue that Algorithm~\ref{alg:decision-algorithm} can be viewed as a combinatorial implementation of a multiplicative weights update algorithm for solving a modified version of the \lovasz SDP, or a combinatorial implementation of the exponentiated gradient descent algorithm applied to the dual of this SDP. In fact, these algorithms were our starting point for the design of Algorithm~\ref{alg:decision-algorithm}. The modified version of the \lovasz SDP that is of interest to us is due to Szegedy~\cite{szegedy1994mod-lovasz} and is formulated as follows:}
\begin{equation}\label{opt:mod-lovasz-sdp}
        \begin{aligned}
\hat{\vartheta}(G) \coloneqq \max_{X \in \symm}\quad & \tr(JX) \\
\text{subject to}\quad 
    & \tr(X) = 1,\\
    & X \succeq 0,\\
    & X_{ij} \leq 0,\quad \text{for }\,ij\in E.
\end{aligned}
    \end{equation}
It has been shown in~\cite{dukanovic2007semidefinite} that $\hat{\vartheta}(\bar{G})$ satisfies the same inequalities as the theta number:
$$\omega(G)\leq \hat{\vartheta}(\bar{G})  \leq \chi(G).$$

The multiplicative weights update method is a general algorithmic framework with many applications; see~\cite{arora2012mwu} and the references therein. In particular, it can be used to approximately solve semidefinite programming feasibility problems.
The algorithm takes an SDP feasibility problem and a quantity called the SDP ``width'' as inputs. For a specified accuracy parameter $\delta > 0$, it either outputs that the SDP is infeasible, or returns a $\delta$-feasible point; i.e., a positive semidefinite matrix that violates every linear constraint by at most $\delta$. In the latter case, the SDP may still be infeasible. The algorithm maintains a set of weights associated with the linear constraints, which are updated in each iteration according to an exponential update rule. Moreover, in each iteration, the algorithm calls an oracle that computes the largest eigenvalue and a corresponding eigenvector of a certain symmetric matrix.



In Algorithm~\ref{alg:decision-algorithm}, we replace the exponential update rule by its first-order Taylor approximation. We also approximate the eigenvalue/eigenvector oracle by the \emph{power method}, which gives rise to the connection with walk counting. We show that the ``width'' parameter of the SDP in~\eqref{opt:mod-lovasz-sdp} can be taken to be $n$. Finally, we prove that despite the approximation errors made by the power method, the Taylor expansion, and the multiplicative weights update algorithm itself, we can still correctly decide if there is a clique of size $k$ in a perfect graph.
}

{Algorithm~\ref{alg:decision-algorithm} can also be viewed as a combinatorial implementation of the exponentiated gradient descent algorithm applied to the dual of the modified \lovasz SDP. 
By taking the dual of the SDP in~\eqref{opt:mod-lovasz-sdp} with $\bar{G}$ as input, one can invoke SDP strong duality to show that a perfect graph $G(V,E)$ has no clique of size $k$ if and only if the following optimization problem has a negative optimal value:
\begin{equation}\label{opt:dual-certificate}
\begin{aligned}
    \min_{p\in \Delta_{|\bar{E}| + 1}}\quad& \lambda_\max\left(p_0(J-kI) - \sum_{ij\in \bar{E}} p_{ij} Y_{ij}\right).
\end{aligned}
\end{equation}

The exponentiated gradient descent algorithm is a method for minimizing convex functions over the probability simplex, which is precisely the form of problem~\eqref{opt:dual-certificate}. This algorithm can be seen as a specific instance of the so-called mirror descent method with negative entropy serving as the ``mirror map''; see~\cite{nemirovski1983problem,bubeck2015convex} for details. The algorithm has an exponential update rule, a step size, and needs access to subgradients of the objective function. Algorithm~\ref{alg:decision-algorithm} can be viewed as a combinatorial implementation of this algorithm where the step size is taken to be $\epsilon/n$, the exponential update rule is Taylor expanded to first order, and the computation of the subgradient is approximated by the power method.}




{We note that the fact that Algorithm~\ref{alg:decision-algorithm} has the above connections to approximate versions of first-order methods for semidefinite programming does not imply that it is non-combinatorial. Indeed, as we have shown in Sections~\ref{sec:certificate-for-refuting-clique-of-size-k}--\ref{sec:alg-analysis}, both the algorithm and its proof of correctness can be fully understood without any knowledge of semidefinite programming.}

\section*{Acknowledgments.} The authors are grateful to Maria Chudnovsky and Cemil Dibek for their contributions to this project and many insightful discussions. We also thank Paul Seymour and Noga Alon for their invaluable feedback.




\section{Appendix}

\subsection{Proof of Lemma~\ref{lem:approximation-via-counting-paths}}
\label{app:proof:approximation-via-counting-paths}

\begin{proof}
We first show that \(S\neq\varnothing\). Let the eigenvalues of \(A_G\) be
\(\lambda_1\ge \lambda_2\ge \cdots \ge \lambda_n\ge 0\). Since \(A_G\) is nonzero and positive semidefinite, \(\lambda_1=\lambda_{\max}(A_G)>0\). {By Lemma~\ref{lem:connection-between-walk-counting-and-matrix-multiplication}, we have}
\[
\sum_{i=1}^n \psi_{2\hat T,A_G}(i)
=
\sum_{i=1}^n (A_G^{2\hat T})_{ii}
=
\tr(A_G^{2\hat T})
=
\sum_{i=1}^n \lambda_i^{2\hat T}
>0.
\]
Therefore, there exists some \(i\in[n]\) such that \(\psi_{2\hat T,A_G}(i)>0\).

Let us prove the desired the upper bound. Fixing \(i\in S\), {Lemma~\ref{lem:connection-between-walk-counting-and-matrix-multiplication} and} the Rayleigh quotient inequality give
\[
\frac{\psi_{2\hat T+1,A_G}(i)}{\psi_{2\hat T,A_G}(i)}
=
\frac{(A_G^{2\hat T+1})_{ii}}{(A_G^{2\hat T})_{ii}}
=
\frac{(A_G^{\hat T}e_i)^\top A_G(A_G^{\hat T}e_i)}{(A_G^{\hat T}e_i)^\top (A_G^{\hat T}e_i)}
\le \lambda_{\max}(A_G).
\]
We next prove the desired lower bound. If \(i\notin S\), {then \(
0=\psi_{2\hat T,A_G}(i) = (A_G^{2\hat T})_{ii}=(A_G^{\hat T}e_i)^\top (A_G^{\hat T}e_i),
\)}
which implies \(A_G^{\hat T}e_i=0\). Therefore,
\(
\psi_{2\hat T+1,A_G}(i)=(A_G^{2\hat T+1})_{ii}=(A_G^{\hat T}e_i)^\top A_G(A_G^{\hat T}e_i)=0.
\)
Therefore, we have
\[
\max_{i\in S}\frac{\psi_{2\hat T+1,A_G}(i)}{\psi_{2\hat T,A_G}(i)}
\ge
\frac{\sum_{i\in S}\psi_{2\hat T,A_G}(i)\cdot \frac{\psi_{2\hat T+1,A_G}(i)}{\psi_{2\hat T,A_G}(i)}}{\sum_{i\in S}\psi_{2\hat T,A_G}(i)}
=
\frac{\sum_{i=1}^n \psi_{2\hat T+1,A_G}(i)}{\sum_{i=1}^n \psi_{2\hat T,A_G}(i)}
=
\frac{\sum_{i=1}^n \lambda_i^{2\hat T+1}}{\sum_{i=1}^n \lambda_i^{2\hat T}}.
\]
To bound the last quotient, let
\[
\ell\coloneqq \left|\left\{i\in[n]:\lambda_i>\lambda_1\!\left(1-\frac{\epsilon}{2}\right)\right\}\right|.
\]
We can lower bound the numerator as
\[
\sum_{i=1}^n \lambda_i^{2\hat T+1}\ge \sum_{i=1}^{\ell}\lambda_i^{2\hat T+1}
\ge \lambda_1\!\left(1-\frac{\epsilon}{2}\right)\sum_{i=1}^{\ell}\lambda_i^{2\hat T}.
\]
Since \(\lambda_i\le \lambda_1(1-\epsilon/2)\) for \(i>\ell\), we can upper bound the denominator as
\begin{align*}
    \sum_{i=1}^n \lambda_i^{2\hat T}
&=
\sum_{i=1}^{\ell}\lambda_i^{2\hat T}+\sum_{i=\ell+1}^n \lambda_i^{2\hat T}\\
&\le
\sum_{i=1}^{\ell}\lambda_i^{2\hat T}+(n-\ell)\lambda_1^{2\hat T}\left(1-\frac{\epsilon}{2}\right)^{2\hat T}\\
&\le \left(1+n\left(1-\frac{\epsilon}{2}\right)^{2\hat T}\right)\sum_{i=1}^{\ell}\lambda_i^{2\hat T}.
\end{align*}
{Using the two bounds above, we have
\begin{equation}
    \begin{aligned}
            \max_{i\in S}\frac{\psi_{2\hat T+1,A_G}(i)}{\psi_{2\hat T,A_G}(i)}
&\ge
\frac{1-\frac{\epsilon}{2}}{1+n(1-\frac{\epsilon}{2})^{2\hat T}}\lambda_1\\
&\ge \frac{1-\frac{\epsilon}{2}}{1+\frac{\epsilon}{2}}\lambda_1\\
&\ge (1-\varepsilon) \lambda_1\\
&= (1-\epsilon)\lambda_{\max}(A_G),
    \end{aligned}
\end{equation}
where the second inequality follows from the fact that $\hat{T} \ge \frac{1}{\epsilon}\ln(\frac{2n}{\epsilon})$ and \(\ln(1-\frac{\epsilon}{2})\le -\frac{\epsilon}{2}\), and the third from the fact that for \(x\in[0,1]\), \(\frac{1-x}{1+x}\ge 1-2x\).}

\end{proof}

\subsection{Proof of Lemma~\ref{lem:lower-bound-eigenvalue}}\label{app:proof:lower-bound-eigenvalue}
\begin{proof}
    We first prove that for any $B,C\in \symm$, the following inequality holds:
    $$\lambda_\min(B+C) \geq \lambda_\min(B) + \lambda_\min(C).$$
   To see this, let $x$ be a unit eigenvector corresponding to $\lambda_{\mathrm{min}}(B+C)$. By the definition of the Rayleigh quotient, we have
    \[
    \lambda_\min(B+C) = x\tran (B+C)x = x\tran Bx + x\tran Cx \geq \lambda_\min(B) + \lambda_\min(C).
    \]
    Using this property, we lower bound $\lambda_{\mathrm{min}}(A_{G_p})$ as follows:
    \begin{align*}
        \lambda_{\mathrm{min}}(A_{G_p}) &\geq \lambda_{\mathrm{min}}\left(\sum_{ij\in \bar{E}}-p_{ij}Y_{ij}\right) + \lambda_{\mathrm{min}}(p_0 J) + 1 \\
        &\geq \lambda_{\mathrm{min}}\left(\sum_{ij\in \bar{E}}-p_{ij}Y_{ij}\right) + 1 \\
        &\geq 0,
    \end{align*}
    where the first inequality follows from the claim above. The second inequality holds because $p_0 \geq 0$ and $J \succeq 0$. The final inequality follows from the Gershgorin circle theorem. Since the diagonal entries of $A_{G_p}$ are nonzero, $A_{G_p}$ is nonzero. This concludes the proof.
\end{proof}

\bibliographystyle{unsrt}
\bibliography{refs}

@article{abrishami2025submodular,
  title={Submodular functions and perfect graphs},
  author={Abrishami, Tara and Chudnovsky, Maria and Dibek, Cemil and Vu{\v{s}}kovi{\'c}, Kristina},
  journal={Mathematics of Operations Research},
  volume={50},
  number={1},
  pages={189--208},
  year={2025},
  publisher={INFORMS}
}

@incollection{Grotschel1994poly-alg-perfect-graphs,
  title={Polynomial algorithms for perfect graphs},
  author={Gr{\"o}tschel, Martin and Lov{\'a}sz, L{\'a}szl{\'o} and Schrijver, Alexander},
  booktitle={North-Holland mathematics studies},
  volume={88},
  pages={325--356},
  year={1984},
  publisher={Elsevier}
}

@article{trotignon2013perfect-graphs-survey,
  title={Perfect graphs: a survey},
  author={Trotignon, Nicolas},
  journal={arXiv preprint arXiv:1301.5149},
  year={2013}
}

@book{grotschel2012geometric-alg-combo-opt,
  title={Geometric Algorithms and Combinatorial Optimization},
  author={Gr{\"o}tschel, Martin and Lov{\'a}sz, L{\'a}szl{\'o} and Schrijver, Alexander},
  volume={2},
  year={2012},
  publisher={Springer Science \& Business Media}
}

@article{laurent2012semidefinite-lecture-notes,
  title={Semidefinite optimization},
  author={Laurent, Monique and Vallentin, Frank},
  journal={Lecture Notes, available at \small{https://www.mi.uni-koeln.de/opt/wp-content/uploads/2015/10/laurent_vallentin_sdo_2012_05.pdf}},
  year={2012}
}

@inproceedings{eisenbrand2003combo-max-stable-t-perfect,
  title={A combinatorial algorithm for computing a maximum independent set in a t-perfect graph},
  author={Eisenbrand, Friedrich and Funke, Stefan and Garg, Naveen and K{\"o}nemann, Jochen},
  booktitle={Proceedings of the Fourteenth Annual ACM-SIAM Symposium on Discrete Algorithms},
  volume={12},
  number={14},
  pages={517--522},
  year={2003}
}

@inproceedings{szegedy1994mod-lovasz,
  title={A note on the $\theta$ number of {L}ov\a' asz and the generalized {D}elsarte bound},
  author={Szegedy, Mario},
  booktitle={Proceedings 35th Annual Symposium on Foundations of Computer Science},
  pages={36--39},
  year={1994},
  organization={IEEE}
}

@article{arora2012mwu,
  title={The multiplicative weights update method: a meta-algorithm and applications},
  author={Arora, Sanjeev and Hazan, Elad and Kale, Satyen},
  journal={Theory of Computing},
  volume={8},
  number={1},
  pages={121--164},
  year={2012},
  publisher={Theory of Computing Exchange}
}

@book{nemirovski1983problem,
  title={Problem Complexity and Method Efficiency in Optimization},
  author={Nemirovski, Arkadi Semenovi{\v{c}} and Yudin, David Borisovich},
  year={1983},
  publisher={John Wiley \& Sons},
  series    = {Wiley-Interscience Series in Discrete Mathematics},
}

@article{chudnovsky2015coloring,
  title={Coloring perfect graphs with no balanced skew-partitions},
  author={Chudnovsky, Maria and Trotignon, Nicolas and Trunck, Th{\'e}ophile and Vu{\v{s}}kovi{\'c}, Kristina},
  journal={Journal of Combinatorial Theory, Series B},
  volume={115},
  pages={26--65},
  year={2015},
  publisher={Elsevier}
}

@article{schrijver2000combinatorial,
  title={A combinatorial algorithm minimizing submodular functions in strongly polynomial time},
  author={Schrijver, Alexander},
  journal={Journal of Combinatorial Theory, Series B},
  volume={80},
  number={2},
  pages={346--355},
  year={2000},
  publisher={Elsevier}
}

@misc{lovasz2011interview,
  author       = {Wigderson,Avi and Lov{\'a}sz, L{\'a}szl{\'o}},
  title        = {Science {L}ives: L{\'a}szl{\'o} {L}ov{\'a}sz},
  howpublished = {Interview by Avi Wigderson \href{[Video]}{https://www.youtube.com/watch?v=ikOUHnrNaxA}. YouTube (uploaded by Simons Foundation)},
  month        = apr,
  day          = {20},
  year         = {2015},
  url          = {https://www.youtube.com/watch?v=ikOUHnrNaxA},
  note         = {Accessed: 2025-08-24}
}

@article{berge1961farbung,
  title={F\"{a}rbung von {G}raphen, deren s\"{a}mtliche bzw. deren ungerade {K}reise starr sind},
  author={Berge, Claude},
  journal={Wiss. Z. Martin-Luther-Univ. Halle-Wittenberg Math.-Natur. Reihe},
  volume = {10},
  pages = {114},
  year={1961}
}

@article{chudnovsky2006strong,
  title={The strong perfect graph theorem},
  author={Chudnovsky, Maria and Robertson, Neil and Seymour, Paul and Thomas, Robin},
  journal={Annals of Mathematics},
  pages={51--229},
  year={2006},
  publisher={JSTOR}
}

@article{Mansour2021ChudnovskyInterview,
  author    = {Toufik Mansour},
  title     = {Interview with {M}aria {C}hudnovsky},
  journal   = {Enumerative Combinatorics and Applications},
  volume    = {1},
  number    = {2},
  year      = {2021},
  note      = {Interview \#S3I4},
  url       = {https://ecajournal.haifa.ac.il/Volume2021/ECA2021_S3I4.pdf}
}

@article{dukanovic2007semidefinite,
  title={Semidefinite programming relaxations for graph coloring and maximal clique problems},
  author={Dukanovic, Igor and Rendl, Franz},
  journal={Mathematical Programming},
  volume={109},
  number={2},
  pages={345--365},
  year={2007},
  publisher={Springer}
}

@article{chudnovsky2020maximum,
  title={On the maximum weight independent set problem in graphs without induced cycles of length at least five},
  author={Chudnovsky, Maria and Pilipczuk, Marcin and Pilipczuk, Micha{\l} and Thomass{\'e}, St{\'e}phan},
  journal={SIAM Journal on Discrete Mathematics},
  volume={34},
  number={2},
  pages={1472--1483},
  year={2020},
  publisher={SIAM}
}

@article{IWATA2002203,
title = {A Fully Combinatorial Algorithm for Submodular Function Minimization},
journal = {Journal of Combinatorial Theory, Series B},
volume = {84},
number = {2},
pages = {203-212},
year = {2002},
issn = {0095-8956},
doi = {https://doi.org/10.1006/jctb.2001.2072},
url = {https://www.sciencedirect.com/science/article/pii/S0095895601920726},
author = {Satoru Iwata}
}

@article{gavril1972algorithms,
  title={Algorithms for minimum coloring, maximum clique, minimum covering by cliques, and maximum independent set of a chordal graph},
  author={Gavril, F{\u{a}}nic{\u{a}}},
  journal={SIAM Journal on Computing},
  volume={1},
  number={2},
  pages={180--187},
  year={1972},
  publisher={SIAM}
}

@article{olariu1991optimal,
  title={An optimal greedy heuristic to color interval graphs},
  author={Olariu, Stephan},
  journal={Information Processing Letters},
  volume={37},
  number={1},
  pages={21--25},
  year={1991},
  publisher={Elsevier}
}

@article{penev2012coloring,
  title={Coloring bull-free perfect graphs},
  author={Penev, Irena},
  journal={SIAM Journal on Discrete Mathematics},
  volume={26},
  number={3},
  pages={1281--1309},
  year={2012},
  publisher={SIAM}
}

@incollection{hsu1981color,
  title={How to color claw-free perfect graphs},
  author={Hsu, Wen-Lian},
  booktitle={North-Holland Mathematics Studies},
  volume={59},
  pages={189--197},
  year={1981},
  publisher={Elsevier}
}

@article{chudnovsky2005recognizing,
  title={Recognizing {B}erge Graphs},
  author={Chudnovsky, Maria and Cornu{\'e}jols, G{\'e}rard and Liu, Xinming and Seymour, Paul and Vu{\v{s}}kovi{\'c}, Kristina},
  journal={Combinatorica},
  volume={25},
  number={2},
  pages={143--186},
  year={2005},
  publisher={Springer-Verlag Berlin, Heidelberg}
}

@incollection{karp2009reducibility,
  title={Reducibility among combinatorial problems},
  author={Karp, Richard M.},
  booktitle={50 Years of Integer Programming 1958-2008: from the Early Years to the State-of-the-Art},
  pages={219--241},
  year={2009},
  publisher={Springer}
}

@article{lovasz1979shannon,
  title={On the {S}hannon capacity of a graph},
  author={Lov{\'a}sz, L{\'a}szl{\'o}},
  journal={IEEE Transactions on Information Theory},
  volume={25},
  number={1},
  pages={1--7},
  year={1979},
  publisher={IEEE}
}

@article{lovasz-substitute,
  title={Normal hypergraphs and the perfect graph conjecture},
  author={Lov{\'a}sz, L{\'a}szl{\'o}},
  journal={{D}iscrete Mathematics},
  volume={2},
  number={3},
  pages={253--267},
  year={1972},
  publisher={Elsevier}
}

@article{freund1997decision,
  title={A decision-theoretic generalization of on-line learning and an application to boosting},
  author={Freund, Yoav and Schapire, Robert E.},
  journal={Journal of Computer and System Sciences},
  volume={55},
  number={1},
  pages={119--139},
  year={1997},
  publisher={Elsevier}
}

@article{bubeck2015convex,
  title={Convex {O}ptimization: Algorithms and {C}omplexity},
  author={Bubeck, S{\'e}bastien},
  journal={Foundations and {T}rends in {M}achine {L}earning},
  volume={8},
  number={3-4},
  pages={231--357},
  year={2015},
  publisher={Emerald Publishing limited}
}

\end{document}